\documentclass[pra, a4paper,  showpacs,superscriptaddress]{revtex4-1}
\usepackage{amsmath,amscd,amsfonts,amssymb, amsbsy, color,bbm, bm,amsthm}
\usepackage{enumitem}
\usepackage{graphicx}
\usepackage{amsbsy}

\theoremstyle{plain}  
\newtheorem{thm}{Theorem}
\newtheorem{lem}{Lemma}

\newtheorem{rmk}{Remark}

\usepackage{longtable}

\numberwithin{equation}{section}
\begin{document}
\title{A pedagogical note on the computation of relative entropy \\ of two $n$-mode gaussian states}
\author{K. R. Parthasarathy}\email{krp@isid.ac.in}
\affiliation{Indian Statistical Institute, Theoretical Statistics and Mathematics Unit,Delhi Centre,
7 S.~J.~S. Sansanwal Marg, New Delhi 110 016, India} 
\date{\today}
\begin{abstract}	 
A formula for the relative entropy   $S(\rho\vert\vert\sigma)={\rm Tr}\,\rho\,(\log\,\rho-\log\,\sigma)$ of two  gaussian states $\rho$, $\sigma$ in the boson Fock space $\Gamma(\mathbb{C}^n)$ is presented. It is shown that the relative entropy has a classical and a quantum part: The classical part consists of a weighted linear combination of relative Shannon entropies of $n$ pairs of Bernouli trials arising from the thermal state composition of the gaussian states $\rho$ and $\sigma$. The quantum part has a sum of $n$ terms, that are functions of the  annihilation means and the covariance matrices of 1-mode marginals of the gaussian state  $\rho'$, which is  equivalent to  $\rho$ under a disentangling unitary gaussian symmetry operation of the state $\sigma$. A generalized formula for the Petz-R{\'e}nyi relative entropy $S_\alpha(\rho\vert\vert\sigma)=-\frac{1}{\alpha-1}\log{\rm Tr}\,
\rho^{\alpha}\sigma^{1-\alpha},\ 0<\alpha<1$ for gaussian states $\rho$, $\sigma$ is also presented. Furthermore it is shown that  $S_{\alpha}(\rho\vert\vert\sigma)$ converges to the limit $S(\rho\vert\vert\sigma)$ as $\alpha$ increases to 1.  

\vskip 1.5in 

\begin{center}
	{\Large\bf\em In memory of Robin~L~Hudson}
\end{center}

\end{abstract}

\maketitle
\section{Introduction}

Every gaussian state in $\Gamma(\mathbb{C}^n)$ is completely determined by its annihilation mean and covariance matrix. A finer analysis of the covariance matrix reveals that such a guassian state is composed of a product of $n$ thermal states (including the vacuum state) and an entangling unitary gaussian symmetry operator arising from the group generated by  phase translations and symplectic transformations. Every thermal state, when viewed in a product basis, gives rise to a natural Bernoulli (binomial) trial with success and failure probabilities corresponding to {\em detection} or {\em no-detection} of a particle. 

In this paper we present a formula for the relative entropy $S(\rho\vert\vert\sigma)={\rm Tr}\,\rho\,(\log\,\rho-\log\,\sigma)$ of two gaussian states $\rho,\ \sigma$ in $\Gamma(\mathbb{C}^n)$. It is shown that the relative entropy $S(\rho\vert\vert\sigma)$ contains two distinguished components: the first part is a  weighted linear combination of classical Shannon relative entropies $H(p_{1j}:p_{2j}), j=1,2,\ldots, n$ of $n$ pairs of Bernouli trials arising from the thermal components of the gaussian states $\rho$ and $\sigma$. The second part is purely quantum and it contains a sum of $n$ terms, each of which is arising from the annihilation mean $\alpha_k\in\mathbb{C}$ and the $2\times 2$ covariance matrices $T_k$  of 1-mode marginal states $\rho(\alpha_k,T_k)$ of a gaussian state $\rho'=U^\dag\,\rho\,U$ where $U=W(\pmb{\ell})\,\Gamma(L)$ corresponds to the disentangling gaussian symmetry of the state $\sigma$. 

Extending the above concepts we  derive an explicit formula for the Petz-R\'{e}nyi relative entropy \break $S_\alpha(\rho\vert\vert\sigma)=-\frac{1}{\alpha-1}\log{\rm Tr}\,
\rho^{\alpha}\sigma^{1-\alpha},\ 0<\alpha<1$ for two gaussian states $\rho$, $\sigma$ in $\Gamma(\mathbb{C}^n)$. We show that in the limit $\alpha\rightarrow 1$ the Petz-R\'{e}nyi relative entropy $S_\alpha(\rho\vert\vert\sigma)$ reduces to the relative entropy  $S(\rho\vert\vert\sigma)$ as expected.

\section{Preliminary concepts on gaussian states}

We  focus our attention on  gaussian states in the complex Hilbert space $L^2(\mathbb{R}^n)$ of a quantum system with $n$  degrees of freedom, which constitute a natural extension of the concept of normal (gaussian) distributions  in classical probability~\cite{Par10,Par13, Par-Sen-2015, Tiju-Par-2019, Tiju-Par-2021}.  

We shall use the Dirac notation  $\langle u\vert v\rangle$ for the scalar product  of any two elements $u$ and $v$ of a complex Hilbert space. The scalar product $\langle u\vert v\rangle$ is assumed to be  linear in $v$ and conjugate linear in $u$. For any operator
$A$   and elements $u$, $v$ in the Hilbert space  we write $\langle u\vert A\vert v\rangle = \langle A^\dagger\, u\vert v\rangle$ whenever $v$ is in the domain of  $A$
and $u$ is in the domain of $A^\dagger$, provided they are well-defined. 

Consider the Hilbert space $L^2(\mathbb{R}^n)$, or equivalently, the  boson Fock space $\mathcal{H}=\Gamma(\mathbb{C}^n)$,  over the  complex Hilbert space $\mathbb{C}^n$ of finite dimension $n$.  
Fix a canonical orthonormal basis 	$\left\{ e_k, \ 1\leq k\leq n\right\}$ in  $\mathbb{C}^n$ with 
\begin{eqnarray*}
	e_k=\left( 0,0,\ldots, 0,1,0,\ldots , 0\right)^T
\end{eqnarray*}  
where '1' appears in the $k^{\rm th}$ position.  

At every element $\mathbf{u}\in\mathbb{C}^n$ we associate a pair of operators $a(\mathbf{u})$, $a^\dag(\mathbf{u})$,  called annihilation, creation operators, respectively in the boson Fock space $\mathcal{H}$ and 
\begin{equation}
W(\mathbf{u})=e^{a^\dag(\mathbf{u})-a(\mathbf{u})}
\end{equation}
denotes the {\em Weyl displacement operator} or simply {\em Weyl operator} in $\mathcal{H}$. 
For any state $\rho$ in $\mathcal{H}$ the complex-valued function on $\mathbb{C}^n$ given by 
\begin{equation}
\hat{\rho}(\mathbf{u})={\rm Tr}\, W(\mathbf{u})\,\rho,\ \ \  \mathbf{u}\in \mathbb{C}^n
\end{equation}
is the {\em quantum characteristic function} of $\rho$ at $\mathbf{u}$. 
The state $\rho$ is said to have finite second moments if 
\begin{equation}
\langle a^\dag(\mathbf{u})\,a(\mathbf{u})\rangle_\rho={\rm Tr}\, a^\dag(\mathbf{u})\,a(\mathbf{u})\,\rho\, <\infty, \ \  \forall\ \mathbf{u}\in \mathbb{C}^n. 
\end{equation}
The {\em annihilation mean}, or simply the mean $\mathbf{m}\in\mathbb{C}^n$  in the state $\rho\in\mathcal{H}$ is defined through 
\begin{equation}
\langle a(\mathbf{u})\rangle_\rho={\rm Tr}\, a(\mathbf{u})\,\rho  = \langle\mathbf{u}\vert\mathbf{m}\rangle, \ \  \mathbf{u}\in \mathbb{C}^n. 
\end{equation}
Then 
\begin{equation}
\langle a^\dag(\mathbf{u})\rangle_\rho=\langle\mathbf{m}\vert\mathbf{u}\rangle.  
\end{equation}
Define the observables 
\begin{eqnarray}
q(\mathbf{u})&=&\frac{a(\mathbf{u})+a^\dag(\mathbf{u})}{\sqrt{2}}  \\ 
p(\mathbf{u})&=&\frac{a(\mathbf{u})-a^\dag(\mathbf{u})}{i\,\sqrt{2}}.
\end{eqnarray}
Then the quantum characteristic function of $\rho$ at $\mathbf{u}$ can be expressed as   
\begin{equation}
\hat{\rho}(\mathbf{u})={\rm Tr}\, e^{-i\,\sqrt{2}\, p(\mathbf{u})}\, \rho. 
\end{equation}
Let $\mathbf{u}=\mathbf{x}+i\,\mathbf{y},\ \ \mathbf{x},\, \mathbf{y}\in\mathbb{R}^n.$
Then  $\mathbf{x}\rightarrow q(\mathbf{x})$, 
$\mathbf{y}\rightarrow p(\mathbf{y})$ 
are {\em position} and {\em momentum} fields obeying the commutation relations 
\begin{equation}
\left[q(\mathbf{x}),\, p(\mathbf{y})\right]=i\, \mathbf{x}^T\mathbf{y}. 
\end{equation}
Then 
$$q_k=q(e_k),\ p_k=p(e_k), \ \ k=1,2,\ldots ,n$$ 
yield the {canonical commutation relations} (CCR)~\cite{Par10}
\begin{equation}
\left[q_j,\ q_k\right]=0,\ \ \left[p_j,\ p_k\right]=0,\ \ \left[q_j,\ p_k\right]= i\, \delta_{jk}.
\end{equation}
Variance of $p(\mathbf{u})$ in the state $\rho$ yields a quadratic form in $\left(x_1,x_2,\ldots , x_n; y_1,y_2,\ldots , y_n\right)^T\in\mathbb{R}^{2n}$ so that 
\begin{eqnarray}
{\rm Var}_{\rho}\,p(\mathbf{x}+i\, \mathbf{y})&=& \left(\begin{array}{l} \mathbf{x},\mathbf{y} \end{array}\right)^T\, C\, \left(\begin{array}{ll} \mathbf{x} \\ \mathbf{y}\end{array} \right). 
\end{eqnarray}
Here $C$ is a real $2n\times 2n$ positive definite matrix, satisfying the matrix inequality 
\begin{equation}
C+\frac{i}{2}\,J \geq 0
\end{equation}
with 
\begin{equation}
\label{sym_metric}
J=\left(\begin{array}{cc} 0_n & I_n \\ -I_n & 0_n \end{array}\right)
\end{equation}
being the canonical symplectic matrix in the real symplectic matrix group of order $2n$: 
$$ {\rm Sp}(2n,\mathbb{R})=\left\{L\,\vert\, L^T\,J\,L\,=J\right\}.$$
In  (\ref{sym_metric}) the right hand side  is expressed in the block notation, with $0_n$, $I_n$  denoting $n\times n$ null and identity matrices respectively. 
The correspondence 
\begin{equation}
t\mapsto W(t\,{\mathbf u})=e^{-i\, t\,\sqrt{2}\, p(\mathbf{u})} 
\end{equation}
is a strongly continuous one-parameter group of unitary operators in $t\in\mathbb{R}$, and 
$\hat{\rho}(t,\mathbf{u})={\rm Tr}\, W(t\,{\mathbf u})\, \rho$ 
is {\em the  characteristic function} of the observable $p(\mathbf{u})$, which has the normal or gaussian distribution with mean value 
\begin{equation}
\frac{\langle \mathbf{m}\vert \mathbf{u}\rangle-\langle \mathbf{u}\vert \mathbf{m}\rangle}{i\,\sqrt{2}}
\end{equation}
and variance 
\begin{equation}
\left(\begin{array}{ll} \mathbf{x}^T, & \mathbf{y}^T\end{array}\right)\, C\, \left(\begin{array}{l} \mathbf{x} \\  \mathbf{y}\end{array}\right),\ \ \mathbf{x}+i\,\mathbf{y}=\mathbf{u}.
\end{equation}
Then 
\begin{equation}
\widehat{\rho}\left(t\,(\mathbf{x}+i\,\mathbf{y})\right)={\rm exp}\left[t\, \left(\langle \mathbf{m}\,\vert \mathbf{u}\rangle-\langle \mathbf{u}\vert \mathbf{m}\rangle\right) - t^2\, \left(\begin{array}{ll} \mathbf{x}^T, & \mathbf{y}^T\end{array}\right)\, C\, \left(\begin{array}{l} \mathbf{x} \\  \mathbf{y}\end{array}\right) \right],\ \ 
\end{equation}
for all  $t\in\mathbb{R}$,\ $\mathbf{u}\in\mathbb{C}^n$ with $\mathbf{u}=\mathbf{x}+i\, \mathbf{y}.$ Thus $\rho$ is a quantum gaussian state with mean $\mathbf{m}$ and covariance matrix $C$ if and only if 
\begin{eqnarray}
\widehat{\rho}\left(\mathbf{x}+i\,\mathbf{y}\right)&=&{\rm exp}\left[2\,i\, {\rm Im}(\mathbf{x}-i\,\mathbf{y})^T\,\mathbf{m} -  \left(\begin{array}{ll} \mathbf{x}^T, & \mathbf{y}^T\end{array}\right)\, C\, \left(\begin{array}{l} \mathbf{x} \\  \mathbf{y}\end{array}\right) \right] \ \  \nonumber \\
&=& {\rm exp}\left[2\,i\, \left(\mathbf{x}^T\,{\rm Im}\,\mathbf{m} -\mathbf{y}^T\,{\rm Re}\,\mathbf{m}\right)- \left(\begin{array}{ll} \mathbf{x}^T, & \mathbf{y}^T\end{array}\right)\, C\, \left(\begin{array}{l} \mathbf{x} \\  \mathbf{y}\end{array}\right) \right]\ \ 
\end{eqnarray}
for all $\mathbf{x},\, \mathbf{y}\in\mathbb{R}^n.$ 
We write 
$$\rho(\mathbf{m},C)=n{\rm -mode\  {\rm\ gaussian\  state\  \ }   with\ mean}\ \mathbf{m} \ {\rm and\ covariance \ matrix\ } C.$$
The covariance matrix may be expressed in the $n\times n$ block matrix notation as
\begin{equation}
C=\left(\begin{array}{ll} C_{11} & C_{12} \\ C_{21} & C_{22}\end{array} \right) 
\end{equation} 
where 
\begin{eqnarray}
C_{11}&=&{\rm Cov}\,\left(p_1,p_2,\ldots\, p_n\right) \nonumber \\
C_{22}&=& {\rm Cov}\,\left(q_1,q_2,\ldots\, q_n\right)  \\ 
C_{12}&=&  \left(-{\rm Cov}\, (p_j, q_k)\right) = C_{21}^T.\nonumber 
\end{eqnarray}
Note that 
\begin{equation}
W(\mathbf{z})\, \rho(\mathbf{m},C)\, W(\mathbf{z})^\dag=  \rho(\mathbf{m}+\mathbf{z},C), \ \ \mathbf{z}\in\mathbb{C}^n.
\end{equation} 


\section {Williamson's theorem applied to $n$-mode gaussian covariance matrix} 

Let $C$ be a  covariance matrix of an $n$-mode gaussian state $\rho$ as described above. Then there exist $\nu_1\geq \nu_2\geq \ldots \geq \nu_n \geq 1/2$ and a sympletic matrix 
$L\in {\rm Sp}(2n,\mathbb{R})$ such that 
\begin{equation}
\label{WNF}
C = L^T\, \left(\begin{array}{cc}  {\rm diag}\,(\nu_1, \nu_2,\ldots, \nu_n) & 0_n \\
& \\
0_n &  {\rm diag}\,(\nu_1, \nu_2,\ldots, \nu_n)
\end{array}\right)\, L   
\end{equation}

\begin{thm} To every  $L\in {\rm Sp}(2n,\mathbb{R})$, there exists a unitary operator $\Gamma(L)$ in the boson Fock space $\mathcal{H}=\Gamma(\mathbb{C}^n)$ satisfying 
	\begin{equation}
	\Gamma(L)\, W(\mathbf{u})\,\Gamma(L)^{-1}=W(L\circ \mathbf{u}), \  \mathbf{u}\in \mathbb{C}^n
	\end{equation}
	where  
	\begin{eqnarray}
	L\circ \mathbf{u} &=&  	
	\left(\begin{array}{cc}
	I_n, & i \,I_n 
	\end{array}\right) \, L  \,  \left(\begin{array}{c}
	\mathbf{x} \\ \mathbf{y} 
	\end{array}\right) \nonumber \\
	&=& \left(A_{11}\,  \mathbf{x} + A_{22}\, \mathbf{y}  \right) +i\, \left(A_{21}\,  \mathbf{x} + A_{22}\, \mathbf{y}  \right), 
	\end{eqnarray}
	with  $A_{11},\, A_{12}, \, A_{21},\, A_{22}$ denoting $n\times n$ real matrices constituting  $L$ as   
	\begin{equation}
	L=	\left(\begin{array}{cc}
	A_{11} & A_{12}  \\
	A_{21} & A_{22}
	\end{array}\right), 
	\end{equation}
	and $\mathbf{u}=\mathbf{x}+i\,\mathbf{y}$, $\mathbf{x}, \mathbf{y}\in \mathbb{R}^n$. The unitary operator  $\Gamma(L)$ is unique up to a scalar multiple of modulus unity. 
\end{thm}
\begin{proof}
	Follows as a consequence of Stone-von Neumann theorem on Weyl operator~\cite{Par10,Par13}.
\end{proof}

\section{Structure theorem for $n$-mode gaussian states } 

For $0< s\leq  \infty$, we may associate a single mode gaussian {\em thermal} state $\rho(s)$  in $\Gamma(\mathbb{C})$ by 
\begin{equation}
\rho(s)=(1-e^{-s})\, e^{-s\, a^\dag\, a}
\end{equation} 
where  $s$ denotes inverse temperature  and  $\rho(\infty)=\vert \Omega\rangle \langle \Omega \vert$,  the vacuum state or the thermal state at zero temperature. Every thermal state 
$\rho(s)$ yields a natural binomial distribution with probability for success equal to $e^{-s}$, where success stands for the event that the number of particles detected is greater than or equal to 1. Then, failure is equivalent to {\em no particle count} with probability $1-e^{-s}$.

The von Neumann entropy of the thermal state $\rho(s)$ is given by 
\begin{eqnarray}
S\left(\rho(s)\right)&=&-{\rm Tr} \rho(s)\, \log\, \rho(s) \nonumber \\
&=& \frac{H(e^{-s})}{1-e^{-s}}. 
\end{eqnarray}
Here $H(p)=-p\, \log\, p - (1-p)\, \log\, (1-p),\  0\leq p \leq 1$ is the Shannon function. 
Note that $\rho(s)\rightarrow \rho(\infty)=\vert \Omega\rangle \langle \Omega \vert$ as $s\rightarrow \infty$. The  state $\rho(\infty)$ being pure,  
$S\left(\rho(\infty)\right)=0$.  

Consider any $n$-mode gaussian state $\rho(\mathbf{m}, C)$ with mean $\mathbf{m}\in \mathbb{C}^n$ and real symmetric positive definite $2n\times 2n$ covariance matrix $C$. Then, there exist $0< s_1 \leq s_2 \leq \ldots \leq s_n \leq \infty$, and a symplectic matrix $L\in {\rm Sp}(2n,\mathbb{R})$ such that 
\begin{equation}
\rho(\mathbf{m}, C)=W(\mathbf{m})\, \Gamma(L)\,  \rho(s_1)\otimes\rho(s_2)\otimes \ldots \otimes \rho(s_n) \, \Gamma(L)^{-1}\, W(\mathbf{m})^{-1}. 
\end{equation}
The sequence $s_1,s_2,\ldots , s_n$ is unique. The covariance matrix $C$ is given by  
\begin{eqnarray}
C=\left(L^{-1}\right)^T\, D(\mathbf{s}) L^{-1},\ \ 
D(\mathbf{s})=\left(\begin{array}{cc} D_0({\mathbf{s}}) &  0_n \\ 
0_n & D_0({\mathbf{s}})  \end{array}   \right), \\
D_0({\mathbf{s}})={\rm diag}\left[\frac{1}{2}\coth\left(\frac{s_k}{2}\right), k=1,2,\ldots n\ \right].
\end{eqnarray}

Thus the von Neumann entropy $S(\rho)$ is given by 
\begin{equation}
S(\rho)=\sum_{k:\, s_k<\infty} \frac{H\left(e^{-s_k}\right)}{1-e^{-s_{k}}}.
\end{equation}
\section{Relative entropy $S(\rho\vert\vert\sigma)$ of two gaussian states}

To set the stage ready for the computation of relative entropy $S(\rho\vert\vert\sigma)={\rm Tr}\,\rho\,(\log\,\rho-\log\,\sigma)$ of two gaussian states $\rho$, $\sigma$ we present the following lemmas. 

\begin{lem} 
	\label{L1} Let  $\vert \psi \rangle \in  {\mathcal H}, \ \vert\vert\, \psi\, \vert\vert=1$, and $\rho$ be a state in ${\mathcal H}$ such that $\langle \psi\vert \rho\,\vert \psi\rangle<1$. Then 
	\begin{equation}
	S(\,\rho \left \vert\vert \,   \vert \psi \rangle \langle \psi \vert \right. ):= \infty.
	\end{equation}	
\end{lem}
\begin{proof}
	Write $\vert \psi_0\rangle=\vert \psi \rangle,\,\vert \psi_1\rangle,\, \vert \psi_2\rangle, \ldots  $   such that $\{ \vert \psi_k\rangle,\, k=0,1,\ldots  \} $ forms an orthonormal basis (ONB) in ${\mathcal H}$. Then 
	\begin{equation}
	\sum_{k=0}^\infty \langle \psi_k\vert \rho\,\vert \psi_k\rangle=1.
	\end{equation} 
	For all $k_0: \, k_0\geq 1$ such that $\langle \psi_{k_0}\vert \rho \,\vert\psi_{k_0}\rangle>0,\ \vert \psi_{k_0}\rangle$   is an eigenvector of  
	$-\log\, \vert \psi_0 \rangle \langle \psi_0 \vert$ with eigenvalue $\infty$. Thus  
	$$-{\rm Tr}\, \rho \log\,  \vert \psi_0 \rangle \langle \psi_0\vert \geq \langle \psi_{k_0}\vert \rho\,\vert \psi_{k_0}\rangle \, \times \infty=\infty.  $$ 
	Then
	\begin{equation}
	S(\,\rho \left \vert\vert \,   \vert \psi_0 \rangle \langle \psi_0 \vert \right. ) = {\rm Tr}\, \rho\, \log\, \rho - {\rm Tr}\, \rho\, \log\,  \vert \psi_0 \rangle \langle \psi_0\vert  =\infty.
	\end{equation}
\end{proof}

\begin{lem} 
	\label{L2} 
	Consider  Hilbert spaces ${\mathcal H}_k,\ k=1,2,\ldots , n$ and states $\sigma_k\in{\mathcal H}_k$.  
	Relative entropy of any arbitrary state   $\rho$  with respect to the state $\sigma=\sigma_1\otimes\sigma_2\otimes \ldots    \otimes\sigma_n,$  in ${\mathcal H}_1\otimes{\mathcal H}_2\otimes \ldots    \otimes{\mathcal H}_n$ is given by 
	\begin{equation}
	S\left(\,\rho \left\vert\vert \, \sigma\right.\right)=-S(\rho)-\sum_{k=1}^{n}{\rm Tr} \, \rho_k\, \log \sigma_k
	\end{equation}
	where $\rho_k$ is the $k$-th marginal of the state $\rho$ in  ${\mathcal H}_k$.
\end{lem}
\begin{proof}
	Write $\widetilde{\sigma}_k=I\otimes I\otimes \ldots \otimes I \otimes \sigma_k\otimes I \otimes \ldots \otimes I,$ where $\sigma_k$ appears at the $k$-th position. Then 
	\begin{equation}
	\log\,\sigma=\sum_{k}\,\log\,\widetilde{\sigma}_k 
	\end{equation}
	and 
	\begin{eqnarray}
	{\rm Tr}\, \rho\, \log\,\sigma&=&\sum_{k}\, {\rm Tr}\,\rho\, \log\,\widetilde{\sigma}_k \nonumber \\
	&=& \sum_{k}\, {\rm Tr}\,\rho_k\, \log\,\sigma_k
	\end{eqnarray}
	So 
	\begin{eqnarray}
	S(\,\rho \left \vert\vert \,\sigma\right.)&=&{\rm Tr}\, \rho\, \log\, \rho - {\rm Tr}\, \rho\, \log\, \sigma  \nonumber \\
	&=& -S(\rho)-\sum_k\, {\rm Tr}\, \rho_k\, \log\, \sigma_k.   
	\end{eqnarray}
\end{proof}

\begin{lem} 
	\label{L3} 
	Let $\rho(m, T)$ be a 1-mode gaussian state with annihilation mean $m\in \mathbb{C}$ and $2\times 2$ covariance matrix $T$. Consider a thermal state $\rho(t)=(1-e^{-t})\,e^{-t\,a^\dag\,a }$ with $a,\, a^\dag$ denoting 1-mode annihilation and creation operators respectively.  Then 
	\begin{equation}
	\label{313}
	{\rm Tr}\, \rho(m, T)\, \log\, \rho(t)=\log\, (1-e^{-t})\, - \frac{t}{2}\, \left({\rm Tr}\, T+2\, \vert m\vert^2-1 \right), \ \ {\rm if\ } 0<t<\infty. 
	\end{equation}
\end{lem}
\begin{proof}  It is readily seen that 
	\begin{eqnarray}
	{\rm Tr}\, \rho(m, T)\, \log\, \rho(t)&=& {\rm Tr}\, \rho(m, T)\, \left\{ \log\, (1-e^{-t})\, -t\, a^\dag\, a\, \right\}   \nonumber 
	\\
	&=& \log\, (1-e^{-t})\, -t\, \left\langle \, a^\dag\, a \right\rangle_{\rho(m, T)}. 
	\end{eqnarray}
	Expressing 
	\begin{eqnarray}
	\left\langle \, a^\dag\, a \right\rangle_{\rho(m, T)}&=&\frac{1}{2}\,  \left\langle \, \left(p^2 +q^2-1\right)\, \right\rangle_{\rho(m, T)} \nonumber \\
	&=& \frac{1}{2}\, \left(  
	{\rm Var}\,(p)_{\rho(m, T)}  +{\rm Var}\,(q)_{\rho(m, T)}+  \langle p\,\rangle_{\rho(m, T)}^2+\langle q\,\rangle_{\rho(m, T)}^2-1  \right) \nonumber \\
	&=& \frac{1}{2}\, \left({\rm Tr}\, T + 2\, \left\vert\,\left\langle \frac{q+ip}{\sqrt{2}}\, \right\rangle_{\rho(m, T)}\,  \right\vert^2-1 \right)  \nonumber  \\
	&=& \frac{1}{2}\, \left({\rm Tr}\, T + 2\, \vert m \vert^2 -1 \right)
	\end{eqnarray}
	we obtain (\ref{313}).
\end{proof}

\begin{lem} 
	\label{L4} 
	Let $\rho=\rho(\mathbf{m}, C),\sigma$ be  gaussian states in $\Gamma(\mathbb{C}^n)$, \ where $\mathbf{m}$ is the annihilation mean  i.e., $\langle\, a(\mathbf{u})\rangle_\rho~=~\langle \mathbf{u}\vert \mathbf{m}\rangle$ and $C$ is the covariance matrix of $\rho$. Let the standard structure of the gaussian state $\sigma$ be 
	\begin{equation}
	\label{316}
	\sigma= W(\pmb{\ell})\, \Gamma(L)\, \rho(t_1)\otimes \rho(t_2)\otimes \ldots \otimes \rho(t_n) \,\Gamma(L)^{-1}\, W(\pmb{\ell})^{-1} 
	\end{equation} 
	so that $\pmb{\ell}$ is the annihilation mean of $\sigma$. Here $L\in{\rm Sp}(2n,\,\mathbb{R})$ with 
	$\Gamma(L)\, W(\mathbf{u})\,\Gamma(L)^{-1}=W(L\circ\mathbf{u}), \ \ \forall\, \mathbf{u}\in\mathbb{C}^n$ and $0<t_1\leq t_2\leq \ldots \leq t_n\leq\infty$ as before in the structure theorem.  Then 
	\begin{equation}
	\label{l4st}
	S(\rho\vert\vert\sigma)=S(\rho'\vert\vert\rho(t_1)\otimes\rho(t_2)\otimes\ldots \otimes \rho(t_n) )
	\end{equation}
	where $\rho'\equiv \rho'\left(\mathbf{m}',C'\right)$ is the gaussian state characterized by  
	the annihilation mean 
	\begin{equation}
	\label{mprime}
	\mathbf{m}'=L^{-1}\circ\,(\mathbf{m}-\pmb{\ell})
	\end{equation} 
	and the covariance matrix 
	\begin{equation}
	\label{cprime}
	C'=L^T\,C\,L.
	\end{equation}
\end{lem}

\begin{proof}
	For any unitary $U$ in $\Gamma(\mathbb{C}^n)$ and for $\rho'=U^{\dag}\rho\,U, \  \  \sigma'=U^{\dag}\,\sigma\,U$ we have 
	\begin{eqnarray}
	S(\rho'\vert\vert \sigma')=S(\rho\, \vert\vert \,\sigma), \ \  
	\end{eqnarray}  
	Choose $U=W(\pmb{\ell})\,\Gamma(L)$  to be the gaussian symmetry leading to the standard form (\ref{316}) of $\sigma$ i.e., 
	$$\sigma'=\left(W(\pmb{\ell})\, \Gamma(L)\right)^{-1}\, \sigma\,\, W(\pmb{\ell})\, \Gamma(L)=\rho(t_1)\otimes\rho(t_2)\otimes\ldots \otimes \rho(t_n).$$
	Noting that  
	\begin{eqnarray*}
		W(\mathbf{z})^{-1}\, \rho(\mathbf{m},\,C)\, W(\mathbf{z})&=&\rho(\mathbf{m}-\mathbf{z},\,C),   \  \ \mathbf{z}\in \mathbb{C}^n\ \ \\    
		\Gamma(M)^{-1}\, \rho(\mathbf{m},\,C)\, \Gamma(M)&=&\rho\left(M^{-1}\circ\mathbf{m},\,M^T\,C\,M\right), \ \   M\in {\rm Sp}(2n,\mathbb{R})
	\end{eqnarray*}
	we recover (\ref{mprime}), (\ref{cprime}) respectively for the annihilation mean and the covariance matrix of the transformed state $\rho'(\mathbf{m}',C')=\left(W(\pmb{\ell})\, \Gamma(L)\right)^{-1}\,\rho\,\left(W(\pmb{\ell})\, \Gamma(L)\right)$. Hence  (\ref{l4st}) follows.    
\end{proof}

Let us denote 
\begin{eqnarray*}
	\pmb{m}_{\rho}&=&L^{-1}\circ (\mathbf{m}-\pmb{\ell})=((m_\rho)_1,(m_\rho)_2,\ldots,(m_\rho)_n)^T\\
	T_k&=&{\rm\  the}\  k{\rm\!-\!th} \ 1\!-\!{\rm mode\ covariance\ matrix\ of\ order\ 2\ in}\  L^T\,C\,L.
\end{eqnarray*}
The 1-mode covariance matrix $T_k$ is made up of the $(kk)$-th, $(k,k+n)$-th, $(k+n,k)$-th, $(k+n,k+n)$-th elements of the transformed $2n\times 2n$ covariance matrix $C'=L^T\,C\,L$.  

\bigskip

Following observations are useful for computing  relative entropy $S(\rho\vert\vert \sigma)$ of $\rho$ with respect to $\sigma$.   

\bigskip 
\noindent{\bf Observation 1}: The $k$-th mode marginal of the transformed state $\rho'(\mathbf{m}',C')=\left(W(\pmb{\ell})\,\Gamma(L)\right)^{-1}\,\rho\, W(\pmb{\ell})\, \Gamma(L)$   is the 1-mode   gaussian state denoted by $\rho((m_{\rho})_k,\, T_k)$.  

\bigskip 

\noindent{\bf Observation 2}: Using Lemma~\ref{L4} we have
$S(\rho\vert\vert\sigma)=S(\rho'\vert\vert\rho(t_1)\otimes\rho(t_2)\otimes\ldots \otimes \rho(t_n) )$.  Since $\rho'$ and $\rho$ are unitarily equivalent it follows that  $S(\rho')=S(\rho)$.



\bigskip 	
\noindent {\bf Observation 3}:  By Lemma~\ref{L2} it follows that  
$$S(\rho\vert\vert \sigma)=-S(\rho)-\sum_{k}\,{\rm Tr}\, \rho((m_\rho)_k,\,T_k)\, \log\,\rho(t_k).$$
If $t_k=\infty$, then  $\rho(t_k)=\vert \Omega\,\rangle \langle \Omega\,\vert$ (vacuum state). If, in addition, 
$\rho((m_\rho)_k,\,T_k)\neq \vert \Omega\,\rangle \langle \Omega\,\vert$, then by Lemma~\ref{L1} we have 
$$-{\rm Tr}\, \rho((m_\rho)_k,\,T_k)\, \log\,\rho(t_k)=\infty.$$ 
Since  
$-{\rm Tr}\, \rho((m_\rho)_k,\,T_k)\, \log\,\rho(t_k)> 0\, \ \forall \ j,$ it follows that $S(\rho\vert\vert \sigma)=\infty$.  

\bigskip 

\noindent{\bf Observation 4}:   If $\rho(t_k)= \vert \Omega\,\rangle \langle \Omega\,\vert, \ {\rm  whenever}\  t_k=\infty,$ 
and by Lemma~\ref{L3}  
\begin{eqnarray*}
	{\rm Tr}\, \rho((m_\rho)_k,\,T_k)\, \log\,\rho(t_k)=\log\, (1-e^{-t_k})-\frac{t_k}{2}\, \left(\, {\rm Tr}\,T_k+2\, \vert (m_\rho)_k\vert^2-1 \right), \ \ {\rm if}\ t_k<\infty. 		
\end{eqnarray*}
we have  
\begin{eqnarray*}
	S(\rho\vert\vert \sigma)&=& -S(\rho) + \sum_{k:\,t_k<\infty}\, \left\{-\log\, (1-e^{-t_k}) +\frac{t_k}{2}\, \left[{\rm Tr}\, T_k +2\, \vert (m_\rho)_k\vert^2 -1   \right]     \right\}.
\end{eqnarray*}
Thus we obtain 
\begin{eqnarray}
S(\rho\vert\vert \sigma)&=&	\left\{\begin{array}{l} 
\infty, \ \  {\rm if\ } \exists\  k \ {\rm such\ that\ } t_k=\infty \, {\rm\ and\ } \rho((m_\rho)_k,T_k)\neq \vert \Omega\rangle\langle \Omega\vert, \\ 
\\
-S(\rho)+  \displaystyle\sum_{k:\,t_k<\infty}\, \left(-\log\, (1-e^{-t_k}) +\frac{t_k}{2}\, \left({\rm Tr}\, T_k +2\, \vert (m_\rho)_k\vert^2 -1   \right)\right), \ \ {\rm otherwise.} 
\end{array} \right.  		
\end{eqnarray}

Let $0<s_1\leq s_2\leq \ldots \leq s_n \leq \infty$, $t_1\leq t_2\leq \ldots \leq t_n \leq \infty$ and $S(\rho\vert\vert \sigma)<\infty$. Then 
\begin{eqnarray}
\label{re1}
S(\rho\vert\vert \sigma)&=& -S(\rho) + \sum_{k=1}^{n}\, \left\{-\log\, (1-e^{-t_k}) +\frac{t_k}{2}\, \left[{\rm Tr}\, T_k +2\, \vert (m_\rho)_k\vert^2 -1   \right] \right\}\nonumber  \\
&=& \sum_{k=1}^{n}\, \left\{ \frac{-H(e^{-s_k})}{(1-e^{-s_k)}} -\log\,(1-e^{-t_k})+\frac{t_k}{2}\, \left[{\rm Tr}\, T_k +2\, \vert (m_\rho)_k\vert^2 -1   \right] \right\}.
\end{eqnarray} 
Drop the suffix $k$ for the $k$-th term of the summation in  (\ref{re1}) for convinience of computation and express it as  
\begin{equation}
\label{star}
\frac{-H(e^{-s})}{(1-e^{-s})} -\log\,(1-e^{-t})+\frac{t}{2}\, \left({\rm Tr}\, T +2\, \vert m_{\rho}\vert^2 -1   \right).
\end{equation}
The first two terms of (\ref{star}) can be expressed as 
\begin{eqnarray}
\frac{-H(e^{-s})}{(1-e^{-s})} -\log\,(1-e^{-t})
&=& \frac{H(p_1:p_2)}{(1-e^{-s})}-\frac{t\,e^{-s}}{(1-e^{-s})}
\end{eqnarray}	
where 
\begin{eqnarray}
H(p_1:p_2)=p_1\,\log\,p_1+(1-p_1)\,\log\,(1-p_1)-p_1\,\log\,p_2-(1-p_1)\,\log\,(1-p_2), \ \ p_1=e^{-s},\ p_2=e^{-t}  
\end{eqnarray}
is the classical Shannon relative entropy of a binomial trial with probability of success $p_1$ with respect to another with probability of sucess $p_2$.   
Thus (\ref{star}) takes the form 
\begin{eqnarray}
\label{star1}
\frac{H(e^{-s}: e^{-t})}{(1-e^{-s})}+\frac{t}{2}\, \left[{\rm Tr}\, T -\left(1+\frac{2\,e^{-s}}{1-e^{-s}}  \right)+2\,  \vert m_{\rho}\vert^2    \right] 
&=& \frac{H(e^{-s}: e^{-t})}{(1-e^{-s})}+\frac{t}{2}\, \left[{\rm Tr}\, T  -\coth\left(\frac{s}{2}\right) +2\, \vert m_{\rho}\vert^2    \right].
\end{eqnarray}
Substituting (\ref{star1})  in (\ref{re1}) leads to 
\begin{eqnarray}
\label{reE}
S(\rho\vert\vert \sigma)=\displaystyle\sum_{k=1}^{n}\, \left\{ \frac{H(e^{-s_k}: e^{-t_k})}{(1-e^{-s_k})}+\frac{t_k}{2}\, \left[{\rm Tr}\, T_k  -\coth\left(\frac{s_k}{2} \right)+2\, \vert (m_\rho)_k\vert^2    \right]  \right\}
\end{eqnarray}  
whenever $S(\rho\vert\vert \sigma)<\infty.$ 
\begin{rmk}
	Relative entropy $S(\rho\vert\vert \sigma)$ of two gaussian states $\rho$ and $\sigma$ consists of a classical part  $\displaystyle\sum_{j=1}^{n}\,\frac{H(e^{-s_k}: e^{-t_k})}{(1-e^{-s_k})}$ and a quantum part $\displaystyle\sum_{k=1}^{n}\,\frac{t_k}{2}\, \left[{\rm Tr}\, T_k-\coth\left(\frac{s_k}{2}\right) +2\, \vert (m_\rho)_k\vert^2      \right].$ 
\end{rmk}

\begin{rmk}
	Note that $\coth\left(\frac{s_k}{2}\right)$ is equal to the trace of the $2\times 2$ covariance matrix of the single mode gaussian state $\rho(s_k)$. 
\end{rmk}	 

\begin{rmk}	 
	Note that 
	\begin{eqnarray*}
		e^{-s_k}&=&{\rm Success\ probability}  \\
		&=& {Pr}(\# {\rm \ particles \ }\geq 1\  {\rm \ in \ the \ thermal \  state}\ \rho(s_k)), \\    
		e^{-t_k}&=&{\rm Success\ probability}  \\
		&=&{Pr}(\# {\rm \ particles \ }\geq 1\  {\rm \ in \ the \ thermal \  state}\  \rho(t_k))
	\end{eqnarray*}
	whereas $1-e^{-s_k}$ and  $1-e^{-t_k}$ are respectively the probabilities of the  "no particle count" in the thermal states $\rho(s_k)$ and $\rho(t_k)$. 
\end{rmk}	 

\section{Petz-R{\'e}nyi relative entropy $S_\alpha(\rho\vert\vert\sigma),\ 0< \alpha<1$ of two gaussian states}

Consider the $\alpha$-dependent Petz-R{\'e}nyi relative entropy~\cite{Renyi1961,Petz1986} between two  states $\rho,\ \sigma$: 
\begin{eqnarray}
S_\alpha(\rho\vert\vert\sigma)=-\frac{1}{\alpha-1}\log{\rm Tr}\,
\rho^{\alpha}\sigma^{1-\alpha}  
\end{eqnarray}
for $0<\alpha<1$. Let $\rho$ and $\sigma$ be two $n$-mode gaussian states with their respective standard forms given by  
\begin{eqnarray}
\rho&=&W(\bm{\ell})\, \Gamma(L)\,  \rho(s_1)\otimes\rho(s_2)\otimes \ldots \otimes \rho(s_n) \, \left(W(\bm{\ell})\, \Gamma(L)\right)^{-1} \nonumber \\ 
\sigma&=&W(\mathbf{m})\, \Gamma(M)\,  \rho(t_1)\otimes\rho(t_2)\otimes \ldots \otimes \rho(t_n) \, \left(W(\mathbf{m})\, \Gamma(M)\right)^{-1}
\end{eqnarray}
where $\bm{\ell},\mathbf{m}\in\mathbb{C}^{n};$  $L, \ M \in$ Sp$(2n,\,\mathbb{R})$; $0\leq s_1 \leq s_2\leq \ldots s_n\leq \infty$; $0\leq t_1 \leq t_2\leq \ldots t_n\leq \infty$. Let 
$U=W(\mathbf{m})\, \Gamma(M)$. Then $\rho'=U^{-1}\,\rho\, U$ and $\sigma'=U^{-1}\,\sigma\, U$ are gaussian states whose standard forms are given by 
\begin{eqnarray}
\rho'&=&W(M^{-1}\circ(\bm{\ell}-\mathbf{m}))\, \Gamma(M^{-1}\,L)\,  \rho(s_1)\otimes\rho(s_2)\otimes \ldots \otimes \rho(s_n) \, \left(W(M^{-1}\circ(\bm{\ell}-\mathbf{m}))\, \Gamma(M^{-1}\,L)\,\right)^{-1} \nonumber \\ 
\sigma'&=&  \rho(t_1)\otimes\rho(t_2)\otimes \ldots \otimes \rho(t_n).
\end{eqnarray}
and 
\begin{equation}
S_{\alpha}(\rho\vert\vert \sigma)=S_\alpha(\rho'\vert\vert \sigma').
\end{equation}
Thus, in the computation of $S_{\alpha}(\rho\vert\vert \sigma)$ we may and do assume that $\rho$ and $\sigma$ have their standard forms given by 
\begin{eqnarray}
\rho&=&W(\mathbf{m}_\rho)\, \Gamma(L_\rho)\,  \rho(s_1)\otimes\rho(s_2)\otimes \ldots \otimes \rho(s_n) \, \left(W(\mathbf{m}_\rho)\, \Gamma(L_\rho)\,\right)^{-1} \nonumber  \\ 
\sigma&=&  \rho(t_1)\otimes\rho(t_2)\otimes \ldots \otimes \rho(t_n).
\end{eqnarray}
where the covariance matrix 
\begin{eqnarray}
\label{d0s}
C_\rho&=& L^T_\rho D(\mathbf{s})\,L_\rho,\   \  \   L_\rho\in {\rm Sp}(2n,\,\mathbb{R}) \nonumber \\ 
D(\mathbf{s})&=&\left(\begin{array}{cc} D_0({\mathbf{s}}) &  0_n \\ 
0_n & D_0({\mathbf{s}})  \end{array}   \right) \nonumber \\  D_0({\mathbf{s}})&=&{\rm diag}\left[\frac{1}{2}\coth\left(\frac{s_k}{2}\right), k=1,2,\ldots n\ \right] 
\end{eqnarray}
and 
\begin{eqnarray}
C_\sigma&=&  D(\mathbf{t})=\left(\begin{array}{cc} D_0({\mathbf{t}}) &  0_n \\ 
0_n & D_0({\mathbf{t}})  \end{array}   \right), \nonumber \\  D_0({\mathbf{s}})&=&{\rm diag}\left[\frac{1}{2}\coth\left(\frac{t_k}{2}\right), k=1,2,\ldots n\ \right]  
\end{eqnarray}
We denote 
\begin{eqnarray}
\rho&=&\rho(\mathbf{m}_\rho,\, L^T_\rho\, D\left(\,\mathbf{s})\, L_\rho\right) \nonumber  \\ 
\sigma&=& \rho(\mathbf{0},\,  D(\mathbf{t})).
\end{eqnarray}

Let us assume that $s_k<\infty$ and $t_k\leq\infty$ for all $k=1,2,\ldots, n$.  For $0<\alpha<1$ we have 
\begin{eqnarray}
\rho^\alpha=\frac{p^\alpha(\mathbf{s})}{p(\alpha\,\mathbf{s})}\,\rho(\mathbf{m}_{\rho},\, L^T_\rho\, D\left(\alpha\,\mathbf{s})\, L_\rho\right)  \\
\sigma^{1-\alpha}=\frac{p^{1-\alpha}(\mathbf{t})}{p((1-\alpha)\,\mathbf{t})}\,\rho(\mathbf{0},\,  D((1-\alpha)\,\mathbf{t}))  
 \end{eqnarray}
where  $p(\mathbf{r})=\displaystyle\prod_{k=1}^{n}\,(1-e^{-r_k})$. Note that  $\rho^\alpha$, $\sigma^{1-\alpha}$ are also  gaussian states   
$\rho\left(\mathbf{m}_{\rho},\, L^T_\rho\, D\left(\alpha\,\mathbf{s}\right)\, L_\rho\right),$ $\rho(\mathbf{0}~,~D((1~-~\alpha)\,\mathbf{t}))$  up to multiplication by scalar factors  $\frac{p^\alpha(\mathbf{s})}{p(\alpha\,\mathbf{s})}$, $\frac{p^{1-\alpha}(\mathbf{t})}{p((1-\alpha)\,\mathbf{t})}$ respectively. 

By Wigner's theorem~\cite{Tiju-Par-2019,Tiju-Par-2021} and gaussian integral formula we obtain   
\begin{eqnarray}
&&{\rm Tr}\rho^\alpha\,\sigma^{1-\alpha} = \frac{1}{\pi^n}
\int_{\mathbb{R}^{2n}}\, d\mathbf{x}\,d\mathbf{y}\ \widehat{(\rho^\alpha)}(\mathbf{x}+i\mathbf{y})
\, \widehat{(\sigma)^{1-\alpha}}(\mathbf{x}+i\mathbf{y})  \nonumber \\
&&\ \ \ = \left(\frac{p^\alpha(\mathbf{s})}{p(\alpha\,\mathbf{s})}\right)\,\left( \frac{p^{1-\alpha}(\mathbf{t})}{p((1-\alpha)\,\mathbf{t})}\right) 
\left(\frac{ {\rm exp}\left[-\widetilde{\mathbf{m}}_{\rho}^T\,\left( L^T_\rho\,D(\alpha\,\mathbf{s})\,L_\rho+ D((1-\alpha)\,\mathbf{t})\right)^{-1}\,\widetilde{\mathbf{m}}_{\rho}\right]}{ {\sqrt{\det\left[L^T_\rho\,D(\alpha\mathbf{s})\,L_\rho +D((1-\alpha)\mathbf{t})\right]}}}\right)
\end{eqnarray}
where $\widetilde{\mathbf{m}}_{\rho}^T=(-\,{\rm Im}\,\mathbf{m}_{\rho},\ {\rm Re}\,\mathbf{m}_{\rho})^T$.  We express the Petz-R{\'e}nyi relative entropy $S_\alpha(\rho\vert\vert\sigma),\ 0<\alpha<1$ as
\begin{eqnarray}
\label{prE}
S_\alpha(\rho\vert\vert\sigma) &=& {\mathcal R}^{I}_\alpha + {\mathcal R}^{II}_\alpha +{\mathcal R}^{III}_\alpha + {\mathcal R}^{IV}_\alpha  \\ 
\label{prE1}
{\mathcal R}^{I}_\alpha&=& -\frac{1}{1-\alpha}\,\sum_{k=1}^n\, \left[\alpha\,\log\left(1-e^{-s_k}\right)-\log\left(1-e^{-\alpha\,s_k}\right)\right]  \\ 
\label{prE2}
{\mathcal R}^{II}_\alpha&=& \sum_{k=1}^n\, \left[-\log\left(1-e^{-t_k}\right) +\frac{1}{1-\alpha}\,\log\left(1-e^{-t_k(1-\alpha)}\right)\right]  \\  
\label{prE3}
{\mathcal R}^{III}_\alpha&=& \frac{1}{1-\alpha}\, \left[\widetilde{\mathbf{m}}_{\rho}^T\,  
\left\{L^T_\rho\, D\left(\alpha\,\mathbf{s}\right)\, L_\rho + D((1-\alpha)\,\mathbf{t})\right\}^{-1}\,\widetilde{\mathbf{m}}_{\rho}\right]  \\
\label{prE4}
{\mathcal R}^{IV}_\alpha&=& \frac{1}{2(1-\alpha)}\,\log\, \det\left[L^T_\rho\, D\left(\alpha\,\mathbf{s}\right)\, L_\rho +D((1-\alpha)\mathbf{t})\right] 
\end{eqnarray}
The  Petz-R{\'e}nyi relative entropy $S_\alpha(\rho\vert\vert\sigma)$   given by equations (\ref{prE})-(\ref{prE4}) converges to $S(\rho\vert\vert\sigma)$  in the limit $\alpha\rightarrow 1$. 

By L'Hospital rule we obtain 
\begin{eqnarray}
\label{r1fin}
\lim_{\alpha\rightarrow 1}\,R_\alpha^{I}&=& \sum_{k=1}^n\,\lim_{\alpha\rightarrow 1}  \left[-\frac{\alpha\, \log\left(1-e^{-s_k}\right)}{1-\alpha}\,\,+\frac{\log\left(1-e^{-\alpha\,s_k}\right)}{1-\alpha}\right]   \nonumber \\ 
&=& \sum_{k=1}^{n}\,\left[\log\left(1-e^{-s_k}\right)-\frac{s_k\, e^{-s_k}}{1-e^{-s_k}}\right]\nonumber \\ 
                              &=&-\sum_{k=1}^{n}\,\frac{H(e^{-s_k})}{1-e^{-s_k}}
\end{eqnarray}  
and 
\begin{eqnarray}
\label{r2fin}
\lim_{\alpha\rightarrow 1}\,R_\alpha^{II}&=&\sum_k\,\left[-\log\left(1-e^{-t_k}\right) +\lim_{\alpha\rightarrow 1}\, \frac{1}{1-\alpha}\,\log\left(1-e^{-t_k(1-\alpha)}\right)\right] \nonumber \\ 
&=& \sum_k\,\left[-\log\left(1-e^{-t_k}\right)+  \lim_{\alpha\rightarrow 1}\,\frac{ t_k\, e^{-t_k(1-\alpha)} }{(1-e^{-t_k(1-\alpha)})}\right] \\
&=& \sum_k\,\left[-\log\left(1-e^{-t_k}\right)+  \lim_{\alpha\rightarrow 1}\,\frac{ t^2_k\, e^{-t_k(1-\alpha)} }{(-t_k\,e^{-t_k(1-\alpha)})}\right] \\
&=& \sum_k\,\left[-\log\left(1-e^{-t_k}\right) - t_k\right]. 
\end{eqnarray}  
In order to compute  $\displaystyle\lim_{\alpha\rightarrow 1}\,R_{\alpha}^{III}$ we take the factor $\frac{1}{1-\alpha}$ inside the bracket $\{\cdots\}^{-1}$ in (\ref{prE3}) to obtain 
\begin{eqnarray}
\label{r3}
\lim_{\alpha\rightarrow 1}\,R_\alpha^{III}&=&\left[\widetilde{\mathbf{m}}_{\rho}^T\,  
\lim_{\alpha\rightarrow 1}\left\{(1-\alpha)\, L^T_\rho\, D\left(\alpha\,\mathbf{s}\right)\, L_\rho +(1-\alpha) D((1-\alpha)\,\mathbf{t})\right\}^{-1}\,\widetilde{\mathbf{m}}_{\rho}\right].   
\end{eqnarray}  
Note that any entry in the $2n\times 2n$ matrix $D(\alpha\,\mathbf{s})$ has the form  (see (\ref{d0s}))  $\frac{1}{2}\,  \coth\left(\frac{\alpha\,s_k}{2}\right)$, which converges to  $\frac{1}{2}\,  \coth\left(\frac{s_k}{2}\right)$ as $\alpha$ increases to 1. Thus, the first term inside the bracket $\{\cdots\}^{-1}$ in (\ref{r3}) is zero. 
Any entry  in the second term  $(1-\alpha)\,D((1-\alpha)\,\mathbf{t})$  is of the form 
$$ \left(\frac{1-\alpha}{2}\right)\,  \coth\left(\frac{(1-\alpha)\,t_k}{2}\right)=\left(\frac{1-\alpha}{2}\right)\, \frac{1+e^{t_k(1-\alpha)}}{1-e^{t_k(1-\alpha)}}.$$ 
Expanding $e^{t_k(1-\alpha)}$ in the denominator and applying L'Hospital rule we get 	 
$$\lim_{\alpha\rightarrow 1}\left(\frac{1-\alpha}{2}\right)\, \frac{1+e^{-t_k(1-\alpha)}}{1-e^{-t_k(1-\alpha)}}=\frac{1}{t_k}.$$   
Thus,
\begin{eqnarray}
\label{r3fin}
\lim_{\alpha\rightarrow 1}\,R_\alpha^{III}&=&\widetilde{\mathbf{m}}_{\rho}^T\, {\rm diag}\left(t_1,\,t_2\,\ldots , t_n; t_1,\,t_2\,\ldots , t_n \right)\, \widetilde{\mathbf{m}}_{\rho}\nonumber \\
    &=& \sum_{k=1}^{n}\, t_{k}\, \vert (m_\rho)_k\,\vert^2
\end{eqnarray}  
To determine the limit $\alpha\rightarrow 1$ of the term ${\mathcal R}^{IV}_\alpha$ (see (\ref{prE4})) we pull out the diagonal $2n\times 2n$ matrix $D((1-\alpha)\mathbf{t})$ from the determinant to obtain  
\begin{eqnarray}
\label{r4d}
R_\alpha^{IV}&=&\frac{1}{2(1-\alpha)}\,\left\{\log\, \det (D((1-\alpha)\mathbf{t}))   +\log\,\det \left(I_{2n}+B_\alpha\,L^T_\rho\, D\left(\alpha\,\mathbf{s}\right)\, L_\rho\, B_\alpha \right)\right\} 
\end{eqnarray}  
where 
\begin{equation}
\label{balpha}
B_\alpha=D^{-\frac{1}{2}}((1-\alpha)\mathbf{t}).
\end{equation}
 Substituting  $\det (D((1-\alpha)\mathbf{t}))=  2\,\displaystyle\sum_{k=1}^n\log \left[\frac{1}{2}\, \coth\left(\frac{(1-\alpha)\,t_k}{2}\right)\right]$ in (\ref{r4d}) and simplifying,  we obtain 
\begin{eqnarray}
\label{r41}
R_\alpha^{IV}&=&\frac{1}{2(1-\alpha)}\,\left\{\sum_{k=1}^n\, 2\, \left[\log\,\left(\frac{1+e^{-t_k(1-\alpha)}}{2}\right)-\log\left(1-e^{-t_k(1-\alpha)}\right)\right]    +\log\,\det \left(I_{2n}+B_\alpha\,L^T_\rho\, D\left(\alpha\,\mathbf{s}\right)\, L_\rho\, B_\alpha \right)\right\}.\nonumber \\ 
\end{eqnarray} 
Applying  L'Hospital rule we obtain  
\begin{eqnarray}
\label{r4t1}
\lim_{\alpha\rightarrow 1}\,\frac{1}{(1-\alpha)}\,\log\,\left(\frac{1+e^{-t_k(1-\alpha)}}{2}\right)&=&\lim_{\alpha\rightarrow 1}\, \frac{-\, t_k\,e^{-t_k(1-\alpha)}}{(1+e^{-t_k(1-\alpha)})}  
=-\frac{t_k}{2}  \\
\label{r4t2}
\lim_{\alpha\rightarrow 1}\,\frac{-1}{(1-\alpha)}\,\log\left(1-e^{-t_k(1-\alpha)}\right)&=& t_k.
\end{eqnarray}
To facilitate the computation of   $\displaystyle\lim_{\alpha\rightarrow 1}\,\frac{1}{2(1-\alpha)}\, \log\,\det \left(I_{2n}+B_\alpha\,L^T_\rho\, D\left(\alpha\,\mathbf{s}\right)\, L_\rho\, B_\alpha \right)$ we present the following Lemma. 
 \begin{lem} 
	\label{L5}
	Let $0<\theta<1$ be a parameter and $F_\theta$ be positive matrices of finite order such that 
	\begin{equation}
	\lim_{\theta\rightarrow 0}\, F_{\theta}=F.  
	\end{equation}   
	Then 
	\begin{equation}
	\label{l5r}
	\lim_{\theta\rightarrow 0}\,\frac{1}{\theta}\,\log\,\det\left(I +\theta\, F_{\theta}\right)={\rm Tr}\,F.  
	\end{equation}   	
\end{lem}
\begin{proof}
The postitive matrices  $F_{\theta}$ admit eigen decomposition 
\begin{equation}
F_{\theta}=\sum_{j}\, \lambda_j(\theta)\, \vert \psi_j(\theta) \rangle \langle  \psi_j(\theta) \vert 
\end{equation}  
using which we express 
	\begin{eqnarray}
\frac{1}{\theta}\,\log\,\det\left(I +\theta\, F_{\theta}\right)&=&   \frac{1}{\theta}\,\log\,\prod_j\,\left(1 +\theta\,\lambda_j(\theta)\right) \\ 
&=& \sum_{j}\,  \frac{1}{\theta}\,\log\,\left(1 +\theta\,\lambda_j(\theta)\right).
\end{eqnarray}   	
As $F_\theta \rightarrow F$ in the limit $\theta\rightarrow 0$, the eigenvalues $\lambda_j(\theta)$ are expected to remain positive and bounded. Expanding $\log(1+\theta\,\lambda_j(\theta))$ in the small $\theta$ limit, up to second order in $\theta$ i.e., 
\begin{equation}
\log\,\left(1 +\theta\,\lambda_j(\theta)\right) \approx \theta\,\lambda_j(\theta) + {\mathcal O}(\theta^2)
\end{equation}
we get 
\begin{eqnarray}
\lim_{\theta\rightarrow 0}\, \sum_{j}\, \frac{1}{\theta}\,\log\,\left(1 +\theta\,\lambda_j(\theta)\right)& =&\lim_{\theta\rightarrow 0} \sum_j \lambda_j(\theta) \nonumber \\ 
&=& {\rm Tr}\,F. 
\end{eqnarray}
\end{proof}
We apply Lemma~\ref{L5} by substituting $\theta=(1-\alpha)$ and $F_{\theta}= \frac{1}{\theta}\,B_{\theta}\,L^T_\rho\, D\left((1-\theta)\,\mathbf{s}\right)\, L_\rho\, B_\theta,$ \  
$B_\theta=D^{-\frac{1}{2}}(\theta\,\mathbf{t})$ to obtain 
\begin{eqnarray}
\label{limr4}
\lim_{\alpha\rightarrow 1}\,\frac{1}{2(1-\alpha)}\, \log\,\det \left(I_{2n}+B_\alpha\,L^T_\rho\, D\left(\alpha\,\mathbf{s}\right)\, L_\rho\, B_\alpha \right)&=&\lim_{\alpha\rightarrow 1}\,\frac{1}{2(1-\alpha)}\,{\rm Tr}\, D^{-1}((1-\alpha)\,\mathbf{t}) \,L^T_\rho\, D\left(\alpha\,\mathbf{s}\right)\, L_\rho.
\end{eqnarray}
The  $2n\times 2n$ diagonal matrix $D((1-\alpha)\,\mathbf{t})$ has entries of the form 
$\frac{1}{2}\,  \coth\left(\frac{1-\alpha}{2}\,t_k\right)$ and we know that 
\begin{eqnarray*}
\lim_{\alpha\rightarrow 1}\,\left(\frac{1-\alpha}{2}\right)\,  \coth\left(\frac{1-\alpha}{2}\,t_k\right)= \lim_{\alpha\rightarrow 1}\,\left(\frac{1-\alpha}{2}\right)\, \frac{1+e^{-t_k(1-\alpha)}}{1-e^{-t_k(1-\alpha)}} =\frac{1}{t_k}.
\end{eqnarray*}
Thus, 
\begin{equation}
\label{dmlim}
\lim_{\alpha\rightarrow 1}\,\frac{1}{(1-\alpha)}\,  D^{-1}((1-\alpha)\,\mathbf{t})= {\rm diag}\,\left(t_1,t_2,\ldots, t_n; t_1,t_2,\ldots, t_n\right)
	\end{equation}
	
Substituting (\ref{dmlim}) in (\ref{limr4}) we get  
\begin{eqnarray}
\label{r4pfin}
\lim_{\alpha\rightarrow 1}\,\frac{1}{2(1-\alpha)}\,   {\rm Tr}\, D^{-1}((1-\alpha)\,\mathbf{t}) \, L^T_\rho\, D(\alpha\,\mathbf{s})\, L_\rho  &=& \frac{1}{2}
 {\rm Tr}\,   {\rm diag}\,\left(t_1,t_2,\ldots, t_n; t_1,t_2,\ldots, t_n\right) C_{\rho} \nonumber \\
 &=& \sum_{k=1}^n  t_k\, {\rm Tr}\,  T_k 
\end{eqnarray}
where $C_{\rho}=L^\rho\, D\left(\mathbf{s}\right)\, L^T_\rho$ is the covariance matrix of the $n$-mode gaussian state $\rho$ and $T_k$ is the $2\times 2$  single mode covariance matrix of the $k$-th mode marginal system of $\rho$. Thus, 
\begin{eqnarray*}
	\label{r4fin}
	\lim_{\alpha\rightarrow 1}\,R_\alpha^{IV}&=&\sum_{k=1}^n\,\frac{t_k}{2} \left({1+\rm Tr}\,  T_k\right)\, 
	\end{eqnarray*}

Putting together  (\ref{r1fin}), (\ref{r2fin}), (\ref{r3fin}) and (\ref{r4fin}) we obtain  
\begin{eqnarray*}
	\lim_{\alpha\rightarrow 1}\,S_\alpha(\rho\vert\vert\sigma)&=& \sum_{k=1}^{n}\, \left\{ \frac{-H(e^{-s_k})}{(1-e^{-s_k)}} -\log\,(1-e^{-t_k})+\frac{t_k}{2}\, \left[{\rm Tr}\, T_k +2\, \vert (m_\rho)_k\vert^2 -1   \right] \right\}\\
	&=& S(\rho\vert\vert\sigma) 
\end{eqnarray*}
  as expected.

For alternate approaches on the computation of relative entropy between two gaussian states see References \cite{Scheel-2001, Chen-2005, Pir-2017,Seshadreesan-Lami-Wilde-2018}.

\section*{Acknowledgement} 
My thanks to A R Usha Devi for making a readable manuscript out of my scribbled algebra.  I thank  Mark M Wilde  for pointing out a misprint in equation (5.8) of the first version
of our manuscript and  Ajit Iqbal Singh for going through the manuscript meticulously.

\providecommand{\MR}{\relax\ifhmode\unskip\space\fi MR }
\providecommand{\MRhref}[2]{%
	\href{http://www.ams.org/mathscinet-getitem?mr=#1}{#2}
}
\providecommand{\href}[2]{#2}


\begin{thebibliography}{DGCZ00}
	
	\section*{References}
	
	
	\bibitem[Chen05]{Chen-2005} X. Y. Chen, {\emph Gaussian relative entropy of entanglement}, Phys. Rev. A. \textbf{71} (2005), 062320.
	
	\bibitem[JP19]{Tiju-Par-2019}  T. C. John, and   K.~R. Parthasarathy, \emph{Klauder-Bargmann Integral Representation of Gaussian Symmetries and Generating Functions of Gaussian States},  arXiv:1911.06555v3 [quant-ph] 
	
	\bibitem[JP21]{Tiju-Par-2021} T. C. John, and  K.~R. Parthasarathy,   \emph{A common parametrization for finite
		mode gaussian states, their symmetries,	and associated contractions with some
		applications} J. Math. Phys. {\bf 62}  (2021), 022102. 
	
	\bibitem[Par10]{Par10}	
	K.~R. Parthasarathy, \emph{What is a gaussian state?}, Commun. Stoch. Anal. \textbf{4}
	(2010),  pp.~143--160. \MR{2662722}
	
	\bibitem[Par13]{Par13}
	K.~R. Parthasarathy, \emph{The symmetry group of {G}aussian states in
		{$L^2(\mathbb{R}^n)$}}, Prokhorov and contemporary probability theory,
	Springer Proc. Math. Stat., {\bf 33}, Springer, Heidelberg, 2013,
	pp.~349--369. \MR{3070484}
	
	\bibitem[PS15]{Par-Sen-2015}
	K.~R. Parthasarathy and Ritabrata Sengupta, \emph{From particle counting to
		gaussian tomography}, Infinite Dimensional Analysis, Quantum Probability and
	Related Topics \textbf{18} (2015),  1550023.
	
	\bibitem[P86]{Petz1986}  D. Petz,
	\emph {Quasi-entropies for finite quantum systems},
	Reports in Mathematical Physics, \textbf{23} (1986), pp.~57--65.
	
	\bibitem[SP17]{Pir-2017} S. Pirandola, R. Laurenza, C. Ottaviani, and
	L. Banchi, {\emph Fundamental Limits of Repeaterless Quantum
		Communications}, Nat. Commun. \textbf{8} (2017), 15043.
	
	
\bibitem[R61]{Renyi1961}  A. R\'enyi,   \emph{On  measures  of  entropy  and  information,} Proceedings of the 4th Berkeley Symposium on Mathematics, Statistics and Probability;  edited by J.~ Neyman,  \textbf{1}, University of California Press, Berkeley, 1961, pp.~547–561. 

	\bibitem[Sch01]{Scheel-2001} S. Scheel, and D. -G. Welsch,  {\emph Entanglement generation and degradation by passive optical devices}, Phys. Rev. A. \textbf{64} (2001), 063811.
	
\bibitem[SLW18]{Seshadreesan-Lami-Wilde-2018} K.~P.~Seshadreesan, L. Lami, and M.~M.~Wilde, \emph{R{\'e}nyi relative entropies of quantum Gaussian states}, J. Math. Phys. \textbf{59}  (2018), 072204.    

	
		
	
\end{thebibliography}
\end{document}